\documentclass[pra,twocolumn,showpacs,superscriptaddress,showkeys]{revtex4}

\usepackage{amsmath,amssymb,amsthm,times,graphics,graphicx,bm}
\usepackage{verbatim}

\newcommand{\trace}{{\rm Tr}}
\newcommand{\adjoint}{\dagger}

\newcommand{\norm}[1]{\left\|\,#1\,\right\|}       
\newcommand{\enorm}[1]{\norm{#1}_{\mathrm{2}}}      
\newcommand{\fnorm}[1]{\norm{#1}_{\mathrm {F}}}    

\newcommand{\set}[1]{{\left\{#1\right\}}}    

\newcommand{\ve}[1]{\mathbf{#1}}
\newcommand{\abs}[1]{\left\lvert #1 \right\rvert}

\newcommand{\complex}{{\mathbb C}}
\newcommand{\reals}{{\mathbb R}}

\newcommand{\nats}{{\mathbb N}}

\newcommand{\discm}{\delta(\rho)}

\def\ket#1{ | #1 \rangle}
\def\bra#1{{\langle #1 | }}
\newcommand{\ketbra}[2]{\ket{#1}\!\bra{#2}}        
\def\tr{ {\rm{Tr }}}
\newcommand{\braket}[2]{\mbox{$\langle #1  | #2 \rangle$}}

\newcommand{\spa}[1]{\mathcal{#1}}
\newcommand{\A}{\spa{A}}
\newcommand{\B}{\spa{B}}
\newcommand{\dens}{\mathcal{D}(\A\otimes\B)}
\newcommand{\denstn}{\mathcal{D}(\complex^2\otimes\complex^N)}
\newcommand{\denstt}{\mathcal{D}(\complex^2\otimes\complex^2)}
\newcommand{\densdd}{\mathcal{D}(\complex^d\otimes\complex^d)}
\newcommand{\herm}{\mathcal{H}(\A\otimes\B)}
\newcommand{\unitaries}{\mathcal{U}(\A\otimes\B)}
\newcommand{\proj}{\set{\Pi_j^A}}
\newcommand{\projv}{\set{\Pi_j^A}_{\ve{v}}}

\newcommand{\rous}{{\rm RU}(\A)}
\newcommand{\UA}{U_A}
\newcommand{\UAv}{U^A_{\ve{v}}}

\newtheorem{theorem}{Theorem}
\newtheorem{lemma}{Lemma}
\newtheorem{cor}[theorem]{Corollary}

\begin{document}
\title{Quantifying non-classicality with local unitary operations}
\author{Sevag Gharibian}
 \affiliation{Institute for Quantum Computing and David R. Cheriton School of Computer Science,\\University of Waterloo, Waterloo, Canada}

\date{\today}
\begin{abstract}
We propose a measure of non-classical correlations in bipartite quantum states based on local unitary operations. We prove the measure is non-zero if and only if the quantum discord is non-zero; this is achieved via a new characterization of zero discord states in terms of the state's correlation matrix. Moreover, our scheme can be extended to ensure the same relationship holds even with a generalized version of quantum discord in which higher-rank projective measurements are allowed. We next derive a closed form expression for our scheme in the cases of Werner states and $(2\times N)$-dimensional systems. The latter reveals that for $(2\times N)$-dimensional states, our measure reduces to the geometric discord [Daki\'{c} et al., PRL 105, 2010]. A connection to the CHSH inequality is shown. We close with a characterization of all maximally non-classical, yet separable, $(2\times N)$-dimensional states of rank at most two (with respect to our measure).
\end{abstract}

\maketitle

\section{Introduction}\label{scn:intro}

One of the most intriguing aspects of quantum mechanics is quantum entanglement, which with the advent of quantum computing, was thrust into the limelight of quantum information theoretic research~\cite{hhhh09}. We now know that correlations in quantum states due to entanglement are necessary in order for \emph{pure-state} quantum computation to provide exponential speedups over its classical counterpart~\cite{jozsa03a}. With bipartite entanglement nowadays fairly well understood, however, attention has turned in recent years to a more general type of quantum correlation, dubbed simply \emph{non-classical correlations}. Unlike entanglement, such correlations \emph{can} be created via Local Operations and Classical Communication (LOCC), but nevertheless do not exist in the classical setting. Moreover, for certain \emph{mixed-state} quantum computational feats, the amount of entanglement present can be small or vanishing, such as in the DQC1 model of computing~\cite{kl98} and the locking of classical correlations~\cite{DHLST04}. In these settings, it is rather non-classical correlations which are the conjectured resource enabling such feats (see, e.g.~\cite{datta05a, datta08a, Luo08, DG08}). In fact, almost all quantum states possess non-classical correlations~\cite{ferraro}.

As a result, much attention has recently been devoted to the quantification of {non-classical correlations} (e.g.,~\cite{PhysRevA.71.062307,gpw05,PhysRevLett.104.080501,groismanquantumness,PhysRevA.77.052101,Luo2008,pianietal2008nolocalbrodcast,pianietal2009broadcastcopies,ADA,PhysRevLett.105.020503,PhysRevA.82.052342,DVB10,streltsov2011,PGACHW11}, see~\cite{MBCPV11} for a survey). Here, we say a bipartite state $\rho$ acting on Hilbert space $\A\otimes \B$ is \emph{classically correlated} in $\A$ if and only if there exists an orthonormal basis $\set{\ket{a}}$ for $\spa{A}$ such that
\[
    \rho = \sum_i p_i \ketbra{a_i}{a_i}\otimes\rho_i
\]
for $\set{p_i}$ a probability distribution and $\rho_i$ density operators. To quantify ``how far'' $\rho$ is from the form above, a number non-classicality measures, including perhaps the best-known such measure, the \emph{quantum discord}~\cite{ollivier01a,henderson01a}, ask the question of how drastically a bipartite quantum state is disturbed under local measurement on $\spa{A}$. In this paper, we take a different approach to the problem. We ask: \emph{Can disturbance of a bipartite system under local unitary operations be used to quantify non-classical correlations?}

It turns out that not only is the answer to this question \emph{yes}, but that in fact for $(2\times N)$-dimensional systems, the measure we construct coincides with the \emph{geometric quantum discord}~\cite{DVB10}, a scheme based again on local measurements. Our measure is defined as follows. Given a bipartite quantum state $\rho$ and unitary $\UA$ acting on Hilbert spaces $\A\otimes\B$ and $\A$ with dimensions $MN$ and $M$, respectively, define
\begin{equation}\label{eqn:measure_U_def}
    D(\rho,\UA) := \frac{1}{\sqrt{2}}\fnorm{\rho - \left(\UA\otimes I_B\right) \rho\left( \UA^\adjoint\otimes I_B\right)},
\end{equation}
where the Frobenius norm $\fnorm{A}=\sqrt{\trace{A^\dagger A}}$ is used due to its simple calculation. Then, consider the set of unitary operators whose eigenvalues are some permutation of the $M$-th roots of unity, i.e. whose vector of eigenvalues equals $\pi\ve{v}$ for $\pi\in S_M$ some permutation and $v_k = e^{2\pi ki/M}$ for $1\leq k \leq M$. We call such operators \emph{Root-of-Unity} (RU) unitaries. They include, for example, the Pauli $X$, $Y$, and $Z$ matrices. Then, letting $\rous$ denote the set of RU unitaries acting on $\A$, we define our measure as:
\begin{equation}\label{eqn:measure_def}
    D(\rho) := \min_{\UA \in\rous} D(\rho,\UA).
\end{equation}
Note that $0\leq D(\rho)\leq 1$ for all $\rho$ acting on $\A\otimes\B$. We now summarize our results regarding $D(\rho)$.\\

\noindent\textbf{Summary of results and organization of paper}

\vspace{1mm}
\noindent\textbf{(A)} Our first result is a closed-form expression for $D(\rho)$ for $(2\times N)$-dimensional systems (Sec.~\ref{scn:twoqubitpure}). This reveals that for $(2\times N)$-dimensional $\rho$, $D(\rho)$ coincides with the geometric discord of $\rho$. It also allows us to prove that, like the \emph{Fu distance}~\cite{f06,gkb08}, if $D(\rho)>1/\sqrt{2}$ for two-qubit $\rho$, then $\rho$ violates the Clauser-Horne-Shimony-Holt (CHSH) inequality~\cite{CHSH69}. The Fu distance, defined as the \emph{maximization} of Eqn.~(\ref{eqn:measure_U_def}) over all $\UA$ such that $[\UA,\trace_{B}(\rho)]=0$, was defined in Ref.~\cite{f06} and studied further in Refs.~\cite{gkb08} and~\cite{DG08} with regards to quantifying entanglement and non-classicality.\\

\vspace{-3mm}
\noindent\textbf{(B)} We next derive a closed form expression for $D(\rho)$ for Werner states (Sec.~\ref{scn:werner}), finding here that $D(\rho)$ in fact equals the Fu distance of $\rho$.\\

\vspace{-3mm}
\noindent\textbf{(C)} Sec.~\ref{scn:purearbdim} proves that only pure maximally entangled states $\rho$ achieve the maximum value $D(\rho)=1$. This is in contrast to the Fu distance, which can attain its maximum value even on non-maximally entangled pure states~\cite{gkb08}.\\

\vspace{-3mm}
\noindent\textbf{(D)} In Sec.~\ref{scn:discord}, we show that $D(\rho)$ is a \emph{faithful} non-classicality measure, i.e. it achieves a value of zero if and only if $\rho$ is classically correlated in $\spa{A}$. To prove this, we first derive a new characterization of states with zero quantum discord based on the correlation matrix of $\rho$. We then show that the states achieving $D(\rho)=0$ can be characterized in the same way. More generally, by extending our scheme to allow the eigenvalues of $\UA$ to have multiplicity at most $k$, we prove a state is undisturbed under $\UA$ if and only if there exists a projective measurement on $\spa{A}$ of rank at most $k$ acting invariantly on the state (Thm.~\ref{thm:gendiscequiv}). This reproduces in a simple fashion a result of Ref.~\cite{MAGGDI11} regarding entanglement quantification in the pure state setting. Based on this equivalence between disturbance under local unitary operations and local projective measurements, we propose a generalized definition of the quantum discord at the end of Sec.~\ref{scn:discord}. In terms of previous work, we note that unlike $D(\rho)$, the Fu distance is not a faithful non-classicality measure~\cite{DG08}. Alternative characterizations of zero discord states have been given in~\cite{ollivier01a,DVB10,D10}. \\

\vspace{-3mm}
\noindent\textbf{(E)} Finally, we characterize the set of maximally non-classical, yet separable, $(2\times N)$-dimensional $\rho$ of rank at most two, according to $D(\rho)$ (and hence according to the geometric discord) (Sec.~\ref{scn:maxnc}). Maximally non-classical separable two-qubit states have previously been studied, for example, in~\cite{GPACH11,GA11}. For example, the set of such states found in Ref.~\cite{GPACH11} with respect to the \emph{relative entropy of quantumness} matches our characterization for $D(\rho)$; we remark, however, that our analysis for $D(\rho)$ in this regard is more general than in~\cite{GPACH11} as it is based on a less restrictive ansatz.\\

\vspace{-3mm}
Sec.~\ref{scn:prelim} begins with necessary definitions and useful lemmas. We conclude in Sec.~\ref{scn:conclusion}. We remark that subsequent to the conception of our scheme, the present author learned that there has also been an excellent line of work studying (the square of) Eqn.~(\ref{eqn:measure_def}) in another setting --- that of \emph{pure state entanglement}. In Ref.~\cite{GI07}, it was found that in $(2\times N)$ and $(3\times N)$ systems, $D(\ketbra{\psi}{\psi})^2$ coincides with the \emph{linear entropy of entanglement}. Ref.~\cite{MAGGDI11} then showed that for arbitrary bipartite pure states, $D(\ketbra{\psi}{\psi})^2$ is a faithful entanglement monotone, and derived upper and lower bounds in terms of the linear entropy of entanglement.

\section{Preliminaries}\label{scn:prelim}
We begin by setting our notation, followed by relevant definitions and useful lemmas. Throughout this paper, we use $\A$ and $\B$ to denote complex Euclidean spaces of dimensions $M$ and $N$, respectively. $\dens$, $\herm$, and $\unitaries$ denote the sets of density, Hermitian, and unitary operators taking $\A\otimes\B$ to itself, respectively. We define $\rho_A := \trace_{B}(\rho)$ and $\rho_B := \trace_{A}(\rho)$, where $:=$ indicates a definition. The Frobenius norm of operator $A$ is $\fnorm{A}=\trace(\sqrt{A^\adjoint A})$, and the anti-commutator of $A$ and $B$ is $\set{A,B}=AB+BA$. The notation $\operatorname{diag}(\ve{v})$ for complex vector $\ve{v}$ denotes a diagonal matrix with $i$th diagonal entry $v_i$, and $\operatorname{span}(\set{\ve{v}_i})$ denotes the span of the set of vectors $\set{\ve{v}_i}$. The minimum (maximum) eigenvalue of Hermitian operator $A$ is denoted $\lambda_{\min}(A)$ ($\lambda_{\max}(A)$), and its $i$th largest eigenvalue is $\lambda_i(A)$. Finally, $\nats$ is the set of natural numbers.

Moving to definitions, in this paper we often decompose $\rho\in\dens$ in terms of a  Hermitian basis for $\herm$ (sometimes known as the Fano form~\cite{F83}):
\begin{eqnarray}\label{eqn:fano}
        \rho = &\frac{1}{MN}&(I^A\otimes I^B + \ve{r}^A\cdot\ve{\sigma}^A\otimes{I^B}+\hspace{10mm}\\&&
                    I^A\otimes\ve{r}^B\cdot\ve{\sigma}^B+\sum_{i=1}^{M^2-1}\sum_{j=1}^{N^2-1}T_{ij}\sigma^A_i\otimes\sigma^B_j).\nonumber
\end{eqnarray}
Here, $\ve{\sigma}^A$ is a $(M^2-1)$-component vector of
traceless orthogonal Hermitian basis elements $\sigma_i^A$ satisfying $\trace(\sigma_i^A\sigma_j^A)=2\delta_{ij}$, $\ve{r}^A\in\reals^{M^2-1}$
is the Bloch vector for subsystem $A$ with
$r^A_i=\frac{M}{2}\tr(\rho_A\sigma^A_i)$, and $T\in\reals^{(M^2-1)\times(N^2-1)}$ is the correlation matrix with entries
$T_{ij}=\frac{MN}{4}\tr(\sigma^A_i\otimes\sigma^B_j\rho)$. For $M=2$, $\ve{r}_A$ satisfies $0\leq\enorm{\ve{r}_A}\leq 1$ with $\enorm{\ve{r}_A} = 1$ if and only if $\rho_A$ is pure. The
definitions for subsystem $B$ are analogous. We now give a useful specific construction for the basis elements $\sigma_i^A$~\cite{he81}. Define
$\set{\sigma_i}_{i=1}^{M^2-1}= \set{U_{pq},V_{pq},W_{r}}$, such
that for $1\leq p<q\leq M$ and $1\leq r \leq M-1$, and
$\set{\ket{i}}_{i=1}^{M}$ some orthonormal basis for $\A$:
    \begin{eqnarray}
        U_{pq}&=&\ket{p}\bra{q}+\ket{q}\bra{p}\label{eqn:Ugenerators}\\
        V_{pq}&=&-i\ket{p}\bra{q}+i\ket{q}\bra{p}\label{eqn:Vgenerators}\\
        W_{r} &=&\sqrt{\frac{2}{r(r+1)}}\!\!\left(\sum_{k=1}^{r}\ket{k}\bra{k}-r\ket{r+1}\bra{r+1}\right).\label{eqn:Wgenerators}
    \end{eqnarray}
Note that when $M=2$, this construction yields the Pauli matrices $\ve{\sigma^A}=(X,Y,Z)$.

Regarding $D(\rho)$, defining $\rho_f := (\UA\otimes I_B) \rho(\UA^\dagger\otimes I_B)$, we often use the fact that Eqn.~(\ref{eqn:measure_def}) can be rewritten as:
\begin{equation}\label{eqn:measure_def2}
    D(\rho) = \min_{U_A\in\rous} \sqrt{\trace(\rho^2) -\trace(\rho\rho_f)}.
\end{equation}

Finally, we show a simple but important lemma.

\begin{lemma}\label{l:invariantlocal}
    $D(\rho)$ is invariant under local unitary operations.
\end{lemma}
\begin{proof}
    Let $\rho' := (V_A\otimes V_B)\rho (V_A\otimes V_B)^\dagger$ for unitaries $V_A$, $V_B$. Then in Eqn.~(\ref{eqn:measure_def2}), $\trace(\rho'^2)=\trace(\rho^2)$, and $\trace(\rho'\rho'_f)$ becomes
    \[
        \trace(\rho (V_A^\dagger\UA V_A\otimes I_B )\rho (V_A^\dagger\UA^\dagger V_A\otimes I_B )).
    \]
   Observe, however, that $V_A\UA V_A^\adjoint$ is still an RU unitary, since we have simply changed basis. Hence, $D(\rho',\UA)=D(\rho,V_{A}^\dagger\UA V_{A})$, and since we are minimizing over all $\UA\in\rous$, the claim follows.
\end{proof}

\section{$(2\times N)$-Dimensional States}\label{scn:twoqubitpure}
In this section, we study $D(\rho)$ for $\rho\in\denstn$, obtaining among other results a closed from expression for $D(\rho)$. To begin, note that any $U_A\in \rous$ must have the form
    \begin{equation}\label{eqn:2qU}
        U_A := \ketbra{c}{c}-\ketbra{d}{d}=2\ketbra{c}{c}-I_2,
    \end{equation}
    up to an irrelevant global phase which disappears upon application of $U_A$ to our system, and for some orthonormal basis $\set{\ket{c},\ket{d}}$ for $\complex^2$. Then, $D(\rho,\UA)$ can be rewritten as
    \begin{equation}\label{eqn:mx2start}
      2\sqrt{\trace[\rho^2 (\ketbra{c}{c}\otimes I) - \rho( \ketbra{c}{c}\otimes I)\rho (\ketbra{c}{c}\otimes I)]}.
    \end{equation}

We begin with a simple upper bound on $D(\rho)$.
\begin{theorem}\label{thm:upperbound1}
     For any $\rho\in \denstn$, one has
     \[
        D(\rho)\leq 2\sqrt{\lambda_{\min}(\trace_{\B}(\rho^2))}.
     \]
\end{theorem}
\begin{proof}
    Starting with Eqn.~(\ref{eqn:mx2start}), by noting that $\trace[\rho( \ketbra{c}{c}\otimes I )\rho (\ketbra{c}{c}\otimes I)]\geq 0$ and using the fact that $\trace(\rho(C_A\otimes I_B))=\trace(\rho_A C_A)$, we have that $D(\rho)$ is at most
    \begin{eqnarray*}
        \min_{\text{unit }\ket{c}\in\complex^2}2\sqrt{\trace[\trace_{\B}(\rho^2)\ketbra{c}{c}]}
               = 2\sqrt{\lambda_{\min}(\trace_{\B}(\rho^2))}.\qedhere
    \end{eqnarray*}
\end{proof}

Thm.~\ref{thm:upperbound1} implies that for pure product $\ket{\psi}\in\complex^2\otimes\complex^N$, $D(\ketbra{\psi}{\psi})=0$, in agreement with the results in Ref.~\cite{GI07}. By next exploiting the structure of $\rho$ further, we obtain a closed form expression for $D(\rho)$.

\begin{theorem}\label{thm:closedform}
    For any $\rho\in \denstn$, define $G:=\ve{r}^A(\ve{r}^A)^T + \frac{2}{N}TT^T$. Then, $D(\rho)$ equals
    \begin{equation}\label{eqn:2qd} \frac{1}{\sqrt{N}}\sqrt{\trace(G)-\lambda_{\max}(G)}=\frac{1}{\sqrt{N}}\sqrt{\lambda_2(G)+\lambda_3(G)}.
    \end{equation}
\end{theorem}
\begin{proof}
Define $P:=\ketbra{c}{c}$. Then, beginning with Eqn.~(\ref{eqn:mx2start}), by rewriting $\rho$ using Eqn.~(\ref{eqn:fano}) and applying the fact that the basis elements $\sigma_i$ are traceless, we obtain that $\trace(\rho^2 P\otimes I-\rho P\otimes I\rho P\otimes I)$ equals
\[
    \frac{1}{4N}\trace(A_1-A_2+A_3-A_4),
\]
where
\begin{eqnarray*}
    A_1 &:=& \left(\sum_i r_i^A{\sigma_i}^A\right)^2P\\
    A_2 &:=& \left(\sum_i r_i^A{\sigma_i}^AP\right)^2\\
    A_3 &:=& \frac{1}{N}\left(\sum_{ij}T_{ij}\sigma^A_i\otimes\sigma^B_j\right)^2(P\otimes I)\\
    A_4&:=&\frac{1}{N}\left(\sum_{ij}T_{ij}\sigma^A_i\otimes\sigma^B_j\right)\left(\sum_{ij}T_{ij}P\sigma^A_iP\otimes \sigma^B_j\right).
\end{eqnarray*}
Using the facts that $(\sigma_i^A)^2=I$, $\set{\sigma^A_i,\sigma^A_j}=0$ for $i\neq j$, $\trace(\sigma_i\sigma_j)=2\delta_{ij}$, and $\trace(P)=1$, we thus have
\begin{eqnarray*}
    \trace(A_1) &=& \enorm{\ve{r}^A}^2,\quad\quad\quad     \trace(A_3) = \frac{2}{N}\sum_{ij} T_{ij}^2\\
    \trace(A_2)&=&\sum_{ij} r_i^Ar_j^A\bra{c}\sigma^A_i\ket{c}\bra{c}\sigma^A_j\ket{c}\\
    \trace(A_4)&=&\frac{2}{N}\sum_{ij} \left(\sum_k T_{ik}T_{jk}\right)\bra{c}\sigma^A_i\ket{c}\bra{c}\sigma^A_j\ket{c}.
\end{eqnarray*}
Now, $\bra{c}\sigma^A_i\ket{c}$ can be thought of as the $i$th component of the Bloch vector $\ve{v}\in\reals^3$ of pure state $\ket{c}$ with $\enorm{\ve{v}}=1$, implying
\[
    \trace(A_2+A_4) = \ve{v}^T\left[\ve{r}^A(\ve{r}^A)^T+\frac{2}{N}TT^T\right]\ve{v}.
\]
Plugging these values into Eqn.~(\ref{eqn:mx2start}), we conclude $D(\rho)$ equals
\[ \min_{\substack{\ve{v}\in\reals^3\\\enorm{\ve{v}}=1}}\frac{1}{\sqrt{N}}\sqrt{\enorm{\ve{r}^A}^2+ \frac{2}{N}\sum_{ij} T_{ij}^2-\trace(A_2+A_4)}.
\]
The claim now follows since for any symmetric $A\in\reals^{n\times n}$,
$\max_{\text{unit }\ve{v}\in\reals^n}\ve{v}^T A\ve{v}=\lambda_{\max}(A)$.
\end{proof}

The expression for $D(\rho)$ in Thm.~\ref{thm:closedform} matches that for the \emph{geometric discord}~\cite{DVB10,VR12}. Specifically, defining the latter as $\delta_g(\rho)=\min_{\sigma\in\Omega} \sqrt{2}\fnorm{\rho-\sigma}$, where $\Omega$ is the set of zero-discord states, we have for $(2\times N)$-dimensional $\rho$ that $D(\rho)=\delta_g(\rho)$. (Note: The original definition of Ref.~\cite{DVB10} was more precisely $\delta_g(\rho)=\min_{\sigma\in\Omega} \fnorm{\rho-\sigma}^2$.) We now discuss consequences of Thm.~\ref{thm:closedform}, beginning with a lower bound which proves useful later.

\begin{cor}\label{cor:lower}
    For $\rho\in \denstn$, we have
    \begin{equation}\label{eqn:22qd} D(\rho)\geq \frac{\sqrt{2}}{N}\sqrt{\lambda_2(TT^T)+\lambda_3(TT^T)}.
    \end{equation}
    This holds with equality if $\ve{r}^A=0$, i.e. $\rho_A=\frac{I}{2}$.
\end{cor}

\begin{proof}
    The first claim follows from the fact that:
    \begin{eqnarray*}
        \lambda_{\max}\left(\ve{r}^A(\ve{r}^A)^T + \frac{2}{N}TT^T\right)\leq\enorm{\ve{r}^A}^2+ \frac{2}{N}\lambda_{\max}\left(TT^T\right).
    \end{eqnarray*}
    The second claim follows by substitution into Eqn.~(\ref{eqn:2qd}).
\end{proof}

For example, for maximally entangled $\ket{\psi}=(\ket{00}+\ket{11})/\sqrt{2}$, for which $\ve{r}^B=\ve{0}$ and $T=\operatorname{diag}(1,-1,1)$, Cor.~\ref{cor:lower} yields $D(\ketbra{\psi}{\psi})= 1$, as desired. We also remark that Eqn.~(\ref{eqn:2qd}) can further be simplified for two-qubit states, since by Ref.~\cite{HH96,HH96_2}, one can assume without loss of generality that $T$ is diagonal. This relies on the facts that (1) applying local unitary $V_1\otimes V_2$ to $\rho$ has the effect of mapping $T\mapsto O_1TO_2^\dagger$, $\ve{r}^A\mapsto O_1\ve{r}^A$, and $\ve{r}^B\mapsto O_2\ve{r}^B$ for some orthogonal rotation matrices $O_1$ and $O_2$, and (2) $D(\rho)$ is invariant under local unitaries by Lem.~\ref{l:invariantlocal}.

Using Cor.~\ref{cor:lower}, we next obtain a connection to the CHSH inequality for two-qubit $\rho$. Defining $M(\rho) := \lambda_1(T^T T)+\lambda_2(T^TT)$, it is known that $\rho$ violates the CHSH inequality if and only if $M(\rho)>1$~\cite{HHH95}. We thus have:

\begin{cor}\label{cor:CHSH}
    For $\rho\in \denstt$, if $D(\rho)>1/\sqrt{2}$, then $M(\rho)>1$. The converse does not hold.
\end{cor}
\begin{proof}
    The first is immediate from Cor.~\ref{cor:lower} and the fact that $TT^T$ and $T^TT$ are cospectral (Thm. 1.3.20 of~\cite{HJ90}). The converse proceeds similarly to Thm. 7 of Ref.~\cite{gkb08} --- namely, let $\ket{\psi}=a\ket{00}+b\ket{11}$ for real $a,b\geq 0$ and $a^2+b^2=1$. Then, for density operator $\ketbra{\psi}{\psi}$, we have $\ve{r}^B=(0,0,a^2-b^2)$ and $T=\operatorname{diag}(2ab,-2ab,1)$, implying $M(\ketbra{\psi}{\psi})>1$ for $a,b\neq 0$. In comparison, $D(\ketbra{\psi}{\psi})=2ab\leq1/\sqrt{2}$ when $a\leq \sqrt{\frac{1}{2}- \frac{1}{2\sqrt{2}}}$ or $a\geq \sqrt{\frac{1}{2}+ \frac{1}{2\sqrt{2}}}$.
\end{proof}
Interestingly, the exact same relationship as that in Cor.~\ref{cor:CHSH} was found between the Fu distance and the CHSH inequality in Ref.~\cite{gkb08}.

\section{Werner States}\label{scn:werner}

We now derive a closed formula for $D(\rho)$ for Werner states $\rho\in \densdd$ where $d\geq 2$, which are defined as~\cite{W89}
\[ \rho:=\frac{2p}{d^2+d}P_{s} + \frac{2(1-p)}{d^2-d}P_a, \]
for $P_s:= (I+P)/2$ and $P_a:=(I-P)/2$ the projectors onto the symmetric and anti-symmetric subspaces, respectively, $P:=\sum_{i,j=1}^d\ketbra{i}{j}\otimes\ketbra{j}{i}$ the SWAP operator, and $0\leq p \leq 1$. Werner states are invariant under $U\otimes U$ for any unitary $U$, and are entangled if and only if $p< 1/2$.
\begin{theorem}\label{thm:werner}
    Let $\rho\in \densdd$ be a Werner state. Then
    \[  D(\rho) = \frac{\abs{2pd-d-1}}{d^2-1}.
    \]
\end{theorem}
\begin{proof}
    As done in Thm. 3 of Ref.~\cite{gkb08}, we first rewrite Eqn.~\ref{eqn:measure_def2} using the facts that $\trace(P)=d$, $\trace(P^2)=d^2$, and $\beta:=\trace(P(U_{A}\otimes I)P(U_{A}\otimes I)^\dagger)=\trace(U_{A})\trace(U_{A}^\dagger)$ to obtain that for any $U_{A}\in U(\A)$,
    \[
        D(\rho, U_{A})=\frac{\sqrt{(2pd-d-1)^2(d^2-\beta)}}{d(d^2-1)}.
    \]
    Since $\trace(U_{A})=0$ for any $U_{A}\in\rous$, we have $\beta=0$ and the claim follows.
\end{proof}

Again, we find that this coincides exactly with the expression for the Fu distance for Werner states~\cite{gkb08}. Further, Thm.~\ref{thm:werner} implies that the quantum discord of Werner state $\rho$ is zero if and only if $p=(d+1)/2d$. This matches the results of Chitambar~\cite{C11}, who develops the following closed formula for the discord $\discm$ of Werner states:
\begin{eqnarray}
    \discm &=&\log (d+1) + (1-p)\log\frac{1-p}{d-1}+p\log\frac{p}{d+1} -\nonumber\\ &&\frac{2p}{d+1}\log p- \left(1-\frac{2p}{d+1}\right)\log \frac{d+1-2p}{2(d-1)}\label{eqn:discWerner}.
\end{eqnarray}
In Sec.~\ref{scn:discord}, we show that this is no coincidence --- it turns out that $D(\rho)= 0$ if and only if the discord of $\rho$ is zero for any $\rho$.

\section{Pure States of Arbitrary Dimension}\label{scn:purearbdim}

We now show that only pure maximally entangled states $\rho$ achieve $D(\rho)=1$. As mentioned in Sec.~\ref{scn:intro}, this is in contrast to the Fu distance~\cite{f06,gkb08}, whose maximal value is attained even for certain \emph{non-maximally} entangled $\ket{\psi}$. We remark that Thm.~\ref{thm:arbdimpure} below also follows from a more general non-trivial result that $D(\ketbra{\psi}{\psi})^2$ is tightly upper bounded by the linear entropy of entanglement of pure state $\ket{\psi}$~\cite{MAGGDI11}. However, our proof of Thm.~\ref{thm:arbdimpure} is much simpler and requires only elementary linear algebra.


To begin, assume without loss of generality that $M\leq N$, and let $\ket{\psi}\in\A\otimes\B$ be a pure quantum state with Schmidt decomposition $\ket{\psi}=\sum_{k=1}^{M}\alpha_k\ket{a_k}\otimes\ket{b_k}$, i.e. $\sum_k\alpha_k^2=1$ for $\alpha_k\in\reals$ and $\set{\ket{a_k}}$ and $\set{\ket{b_k}}$ the Schmidt bases for $\A$ and $\B$, respectively.

\begin{theorem}\label{thm:arbdimpure}
Let $\ket{\psi}\in\A\otimes\B$ with Schmidt decomposition as above. Then $D(\ketbra{\psi}{\psi})=1$ if and only if $\alpha_k=\frac{1}{\sqrt{M}}$ for all $1\leq k \leq M$ (i.e. $\ket{\psi}$ is maximally entangled).
\end{theorem}
\begin{proof}
We begin by rewriting Eqn.~(\ref{eqn:measure_def2}) as
\begin{equation}\label{eqn:arbdimpureD}
    D(\ketbra{\psi}{\psi}) = \min_{\UA\in\rous}\sqrt{1- \abs{\sum_{k=1}^{M}\alpha_k^2\bra{a_k}\UA\ket{a_k}}^2}.
\end{equation}
If $\ket{\psi}$ is maximally entangled, then $\alpha_k = 1/\sqrt{M}$ for all $1\leq k\leq M$. Then, since $\UA\in\rous$, Eqn.~(\ref{eqn:arbdimpureD}) yields
\begin{eqnarray*}
    D(\ketbra{\psi}{\psi}) =
     \min_{\UA\in\rous}\sqrt{1- \frac{1}{M^2}\abs{\trace(\UA)}^2}
    = 1.
\end{eqnarray*}

For the converse, assume $D(\ketbra{\psi}{\psi})=1$. Then, by Eqn.~(\ref{eqn:arbdimpureD}), we must have that for all $\UA\in\rous$,
\begin{equation}\label{eqn:arbdim_pure_proof}
    \sum_{k=1}^{M}\alpha_k^2\bra{a_k}\UA\ket{a_k}=0.
\end{equation}
Thus, choosing $\UA$ as diagonal in basis $\set{\ket{a_k}}$, Eqn.~(\ref{eqn:arbdim_pure_proof}) equivalently says that $\ve{w}^T\pi\ve{v}=0$ for all permutations $\pi\in S_M$, where ${w}_k:=\alpha_k^2$ and ${v}_k := e^{2\pi ki/M}$. This can only hold, however, if all entries of $\ve{w}$ are the same, i.e. $\alpha_k=1/\sqrt{M}$ for all $1\leq k\leq M$, as desired.
\end{proof}

\begin{cor}
    A quantum state $\rho\in\dens$ achieves $D(\rho)=1$ if and only if $\rho$ is pure and maximally entangled.
\end{cor}
\begin{proof}
    Immediate from Thm.~\ref{thm:arbdimpure} and the $\trace(\rho^2)$ in Eqn.~(\ref{eqn:measure_def2}).
\end{proof}

\section{Relationship to Quantum Discord}\label{scn:discord}
We now show that for arbitrary $\rho\in\dens$, $D(\rho)$ is zero if and only if the quantum discord of $\rho$ is zero. The discord is defined as follows~\cite{ollivier01a}:
\begin{equation}\label{eqn:discord_def}
    \discm := S(A)-S(A,B)+\min_{\proj}S(B|\proj),
\end{equation}
where $\proj$ corresponds to a complete measurement on subsystem $B$ consisting of rank $1$ projectors, $S(B)=-\tr(\rho_B\log(\rho_B))$ is the von Neumann entropy of $\rho_B$, similarly $S(A,B)=S(\rho)$, and
\begin{equation}
    S(B|\proj) = \sum_{j}p_jS\left(\frac{1}{p_j}\Pi_j^A\otimes I^B\rho \Pi_j^A\otimes I^B\right),
\end{equation}
where $p_j=\tr(\Pi_j^A\otimes I^B\rho)$. Here, the main fact we leverage about the discord is the following.
\begin{theorem}[Ollivier and Zurek~\cite{ollivier01a}]\label{thm:olzu}
    For $\rho\in\dens$, $\discm=0$ if and only if
    \begin{equation}\label{eqn:disc_exp_proof_eq1}
        \rho = \sum_j \Pi_j^A\otimes I^B\rho \Pi_j^A\otimes I^B,
    \end{equation}
    for some complete set of rank $1$ projectors $\proj$.
\end{theorem}

We now prove the main result of this section. The first part of the proof involves a new characterization of the set of zero discord quantum states $\rho$ in terms of the basis elements $\sigma^A_i$ from the Fano form of $\rho$. Key to this characterization is the absence of non-diagonal $\sigma^A_i$ in the expansion of $\rho$. In the proofs below, we assume the basis elements $\sigma_i^A$ for $\spa{A}$ come from the set $\set{I, U_{pq}, V_{pq},W_{r}}_{p,q,r}^A$ from Sec.~\ref{scn:prelim} (analogously for $\spa{B}$).

\begin{theorem}\label{thm:discequiv}
    Let $\rho\in\dens$. Then $\discm=0$ if and only if there exists a local unitary $V^A$ such that
    \[
        \trace\left(\left(V^A\otimes I^B\right)\rho\left({V^A}^\dagger\otimes I^B\right) \left(\sigma_i^A\otimes \sigma_j^B\right)\right)=0
    \]
    for all $\sigma_i^A\in\set{U_{pq},V_{pq}}^A$ and all $\sigma^B_j\in\set{I, U_{pq}, V_{pq},W_{r}}^B$. The same characterization holds for $D(\rho)=0$.
\end{theorem}
\begin{proof}
We prove the equivalent statement that $\delta(\rho)=0$ if and only if there exists an orthonormal basis $\set{\ket{k}}$ for $\A$ such that, for basis elements $\sigma_i^A$ constructed with respect to $\set{\ket{k}}$, we have $\trace(\rho (\sigma_i^A\otimes \sigma_j^B))=0$ for all $\sigma_i^A\in\set{U_{pq},V_{pq}}$ (and similarly for $D(\rho)=0$).

    Suppose $\discm=0$. Then by Thm.~\ref{thm:olzu}, there exists a complete set of rank 1 projectors $\proj$ such that Eqn.~(\ref{eqn:disc_exp_proof_eq1}) holds. Let $\set{\ket{k}}$ be the basis onto which $\proj$ projects, and define $\Phi(C):=\sum_j\Pi_j^AC\Pi_j^A$. By constructing the basis elements $\sigma_i^A$ in Eqn.~(\ref{eqn:fano}) using $\set{\ket{k}}$, we thus have
    \begin{eqnarray}\label{eqn:disc_exp_proof_eq2}
        \rho &=& \frac{1}{MN}\left[ I^A\otimes I^B + {I^A}\otimes\ve{r}^B\cdot\ve{\sigma}^B+\hspace{10mm}\right.\\&&
                    \left.\sum_{i=1}^{M^2-1}\Phi(\sigma^A_i)\otimes\left(r^A_iI^B+\sum_{j=1}^{N^2-1}T_{ij}\sigma^B_j\right) \right] .\nonumber
    \end{eqnarray}
    Now, for all $\sigma_i^A\in\set{W_r}$, we clearly have $\Phi(\sigma_i^A)=\sigma_i^A$. For $\sigma_i^A\in\set{U_{pq},V_{pq}}$, however, $\Phi(\sigma_i^A)=0$. Thus, in order for Eqn.~(\ref{eqn:disc_exp_proof_eq1}) to hold, we must have $r_i^A=T_{ij}=0$ for all basis elements $\sigma_i^A\in\set{U_{pq},V_{pq}}$, which by definition means $\trace(\rho(\sigma_i^A\otimes \sigma_j^B))=0$ for all $\sigma_i^A\in\set{U_{pq},V_{pq}}^A$, as desired. To show that this implies $D(\rho)=0$, construct $U^A\in\rous$ as diagonal in basis $\set{\ket{k}}$ and define $\Phi(C):=U^AC {U^A}^\dagger$. Then since in Eqn.~(\ref{eqn:disc_exp_proof_eq2}), we have $\Phi(\sigma_i^A)=\sigma_i^A$ for any $\sigma_i^A\in\set{I,W_r}$, the claim follows.

    To show the converse, assume $D(\rho,U^A)=0$ for some $U^A\in\rous$. Then, construct the basis elements $\sigma^A_i$ with respect to a diagonalizing basis $\set{\ket{k}}$ for $U^A$ and define $\Phi(C):=U^AC {U^A}^\dagger$. It follows that for any $p$ and $q$,
    \begin{eqnarray}
        \Phi(U_{pq})&=&e^{i(\theta_p-\theta_q)}\ketbra{p}{q}+e^{-i(\theta_p-\theta_q)}\ketbra{q}{p},\label{eqn:u1}\\
        \Phi(V_{pq})&=&-ie^{i(\theta_p-\theta_q)}\ketbra{p}{q}+ie^{-i(\theta_p-\theta_q)}\ketbra{q}{p}\label{eqn:u2}.
    \end{eqnarray}

    \noindent Consider now an arbitrary term $(c_{u}\sigma^A_u+c_v\sigma^A_v)\otimes\sigma^B_j$ from the Fano form of $\rho$ where $\sigma^A_u=U_{pq}$ and $\sigma^B_v=V_{pq}$ for some choice of $p$ and $q$. Since Eqns.~(\ref{eqn:u1}) and~(\ref{eqn:u2}) imply that $U^A$ can only map $U_{pq}$ to $V_{pq}$ and vice versa, it follows that in order for $D(\rho,U^A)=0$ to hold, we must have
$        \Phi(c_{u}\sigma^A_u+c_v\sigma^A_v)=c_{u}\sigma^A_u+c_v\sigma^A_v.
$
    This leads to the system of equations
    \begin{eqnarray*}
        c_u-ic_v&=&e^{i(\theta_p-\theta_q)}(c_u-ic_v)\\
        c_u+ic_v&=&e^{-i(\theta_p-\theta_q)}(c_u+ic_v).
    \end{eqnarray*}
    We conclude that if either $c_u\neq0$ or $c_v\neq0$, it must be that $\theta_p=\theta_q$ in order for $D(\rho)=0$ to hold. However, since all eigenvalues of $U^A$ are distinct by definition, this is impossible. Thus, $\trace(\rho (\sigma_i^A\otimes \sigma_j^B))=0$ for all $\sigma_i^A\in\set{U_{pq},V_{pq}}$, as desired. To see that this implies $\discm=0$, simply now choose $\proj$ as the projection onto $\set{\ket{k}}$. Then, defining $\Phi(C):=\sum_j\Pi_j^AC\Pi_j^A$ and applying the same arguments from the forward direction to Eqn.~(\ref{eqn:disc_exp_proof_eq2}), we conclude that $\rho$ is invariant under $\proj$. By Thm.~\ref{thm:olzu}, we have $\discm=0$, completing the proof.
\end{proof}

Theorem~\ref{thm:discequiv} shows that $D(\rho)$ defined in Eqn.~(\ref{eqn:measure_def}) is zero precisely for the set of states classically correlated in $\spa{A}$. In other words, unlike the Fu distance~\cite{DG08}, $D(\rho)$ is indeed a \emph{faithful} non-classicality measure. The proof of Thm.~\ref{thm:discequiv} does, however, have a curiosity --- the key property the proof relies on is that all $U^A\in\rous$ have non-degenerate spectra. Interestingly, this is the mixed-state analogue of the pure-state result of Ref.~\cite{MAGGDI11}, where it was shown that a non-degenerate spectrum suffices to conclude $D(\ketbra{\psi}{\psi})$ is a faithful {entanglement monotone} for pure states $\ket{\psi}$. Specifically, Ref.~\cite{MAGGDI11} shows that if in Eqn.~(\ref{eqn:measure_def}) we minimize over $U^A$ with eigenvalues of multiplicity at most $k$ (with at least one eigenvalue of multiplicity $k$), then $D(\ketbra{\psi}{\psi})=0$ if and only if $\ket{\psi}$ has Schmidt rank at most $k$. Could there be an analogue of this more general result in the mixed-state setting of non-classicality? It turns out the answer is yes.

Let $\ve{v}\in\nats^M$ such that $\sum_{j=1}^M v_jj=M$. Then, consider an arbitrary (i.e. not necessarily RU) unitary $\UAv$ which has precisely $v_j$ distinct eigenvalues with multiplicity $j$. For example, $\UAv\in\rous$ has $\ve{v}=(M,0,\ldots,0)$ since it has $M$ distinct eigenvalues of multiplicity $1$. Similarly, if $\ve{v}=(0,0,\ldots,1)$, then $\UAv$ is just the identity (up to phase), and if $\ve{v}=(M-4,2,\ldots,0)$ then $\UAv$ has $M-4$ distinct eigenvalues of multiplicity $1$, and two distinct eigenvalues with multiplicity $2$ each. Now, corresponding to any $\UAv$ is a complete projective measurement $\projv$ which consists precisely of $v_j$ projectors of rank $j$. The correspondence is simple: Let $\lambda$ be an eigenvalue of $\UAv$ with multiplicity $j$, i.e. the projector $\Pi_\lambda$ onto its eigenspace has rank $j$. Then $\Pi_\lambda\in \projv$. It is easy to see that similarly, corresponding to any $\projv$ is a $\UAv$ (assuming we are not concerned with the precise eigenvalues of $\UAv$, as is this case here). We can now state the following.

\begin{theorem}\label{thm:gendiscequiv}
    Let $\rho\in\dens$ and $\ve{v}\in\nats^M$ such that $\sum_{j=1}^M v_jj=M$. Then, there exists a complete projective measurement $\projv$ such that
    \begin{equation}
        \rho = \sum_j \Pi_j^A\otimes I^B\rho \Pi_j^A\otimes I^B\label{eqn:invar}
    \end{equation}
    if and only if there exists a $\UAv\in \mathcal{U}(\spa{A})$ with $D(\rho,\UAv)=0$.
\end{theorem}
\begin{proof}
    The proof follows that of Thm.~\ref{thm:discequiv}, so we outline the differences. Here, $\UAv$ and $\projv$ will be related through the correspondence outlined above, and the basis elements $\sigma_i^A$ are constructed with respect to a diagonalizing basis $\set{\ket{k}}$ for $\UAv$ (which by definition also diagonalizes each $\Pi_j^A\in\projv$). For simplicity, we discuss the case of $\ve{v}=(M-2,1,0,\ldots,0)$; all other cases proceed analogously.

    Going in the forward direction, suppose $\Pi^A_j\in\projv$ projects onto $\spa{S}_{pq}:=\operatorname{span}(\ket{p},\ket{q})$. Then, in Eqn.~(\ref{eqn:disc_exp_proof_eq2}), $\Phi(\sigma_i^A)=\sigma_i^A$ for $\sigma_i^A=U_{pq}$ and $\sigma_i^A=V_{pq}$. In other words, now we can have $r_i^A\neq 0$ and $T_{ij}\neq 0$ (however, note we still have $r_{m\neq i}^A=0$ and $T_{m\neq i,j}=0$). Since $\UAv$ has a degenerate eigenvalue on $S_{pq}$, however, we have by Eqns.~(\ref{eqn:u1}) and~(\ref{eqn:u2}) that $\UAv$ acts invariantly on $\sigma_i^A$ as well (since $\theta_p=\theta_q$). The converse is similar; namely, suppose $\UAv$ has a degenerate eigenvalue on $\mathcal{S}_{pq}$. Then the projector onto the corresponding two-dimensional eigenspace $\Pi^A_j\in\projv$ is $\Pi^A_j=\ketbra{p}{p}+\ketbra{q}{q}$. It thus follows by the same argument as above that both $\UAv$ and $\Pi^A_j$ act invariantly on $U_{pq}$ and $V_{pq}$.
\end{proof}

From this general theorem, we can re-derive as a simple corollary the pure state result of Ref.~\cite{MAGGDI11} mentioned earlier, which we rephrase in our terminology as follows.

\begin{cor}
    Let $\ket{\psi}=\sum_{i=1}^r\alpha_i\ket{\psi^A_i}\ket{\psi^B_i}$ be the Schmidt decomposition of $\ket{\psi}\in\spa{A}\otimes\spa{B}$. Then, there exists $\UAv\in\mathcal{U}(\spa{A})$ with $v_k\geq 1$ (i.e. $\UAv$ has an eigenvalue of multiplicity $k$), $v_{k'>k}=0$ (all eigenvalues of $\UAv$ have multiplicity at most $k$), and $D(\ketbra{\psi}{\psi},\UAv)=0$ if and only if $k\geq r$.
\end{cor}
\begin{proof}
    Suppose $k\geq r$. Then, by defining $\projv^k$ such that $v_k\geq 1$ and $v_{k'>k}=0$, one can choose a $\projv^k$ such that Eqn.~(\ref{eqn:invar}) holds for $\rho=\ketbra{\psi}{\psi}$ (i.e. simply project onto $\operatorname{span}(\set{\ket{\psi^A_i}})$). By Thm.~\ref{thm:gendiscequiv}, this implies there exists a $\UAv$ with $v_k\geq 1$ and $v_{k'>k}=0$ achieving $D(\ketbra{\psi}{\psi},\UA)=0$. Conversely, if $k<r$, then clearly no such $\projv^k$ such that Eqn.~(\ref{eqn:invar}) holds exists. By Thm.~\ref{thm:gendiscequiv}, this implies that no $\UA$ with an eigenvalue of multiplicity at most $k$ and $D(\ketbra{\psi}{\psi},\UA)=0$ exists, as desired.
\end{proof}

We close this section with two final comments. First, given Thm.~\ref{thm:discequiv}, one might ask whether a stronger relationship between $D(\rho)$ and $\discm$ holds. For example, could it be that $D(\rho)\geq\discm$ for all $\rho$? This simplest type of relationship is ruled out easily via Thm.~\ref{thm:werner} and Eqn.~(\ref{eqn:discWerner}), since for $d=2$ and $p=2/3$, $D(\rho)=1/9 \geq \discm \approx 0.01614$, while for $d=50$ and $p=2/3$, $D(\rho)\approx 0.00627 \leq \discm\approx 0.07111$.

Second, note that Thm.~\ref{thm:gendiscequiv} reduces to Thm.~\ref{thm:discequiv} if we choose $\ve{v}=(M,0,\ldots,0)$. This suggests defining a \emph{generalized quantum discord}, denoted $\delta_{\ve{v}}(\rho)$, which is analogous to $\discm$, except that now we use the class of measurements $\projv$ in Eqn.~(\ref{eqn:discord_def}). For example, $\delta_{(M,0,\ldots,0)}(\rho)=\discm$. We hope the study of $\delta_{\ve{v}}(\rho)$ would prove fruitful in its own right.

\section{Maximally Non-Classical Separable States}\label{scn:maxnc}
In this section, we characterize the set of maximally non-classical, yet separable, $(2\times N)$-dimensional states of rank at most $2$, as quantified by $D(\rho)$. To do so, consider separable state
\begin{equation}\label{eqn:sepstate}
    \rho = \sum_{i=1}^{n} p_i \ketbra{a_i}{a_i}\otimes\ketbra{b_i}{b_i},
\end{equation}
where $\sum_i p_i = 1$, $\ket{a_i}\in\complex^2$, $\ket{b_i}\in\complex^N$. Via simple algebraic manipulation, one then finds that $D(\rho,U_A)$ for any given $U_A\in\mathcal{U}(\spa{A})$ is given by
\begin{equation}\label{eqn:sepstateD} \sqrt{\sum_{i=1}^n\sum_{j=1}^np_ip_j\abs{\braket{b_i}{b_j}}^2(\abs{\braket{a_i}{a_j}}^2-\abs{\bra{a_i}U_A\ket{a_j}}^2)}.
\end{equation}
We begin by proving a simple but useful upper bound on $D(\rho)$ which depends solely on $n$.

\begin{lemma}\label{l:upperbound}
    Let $\rho$ be a separable state as given by Eqn.~(\ref{eqn:sepstate}). Then $D(\rho)\leq 1-\max_i p_i \leq1-\frac{1}{n}$.
\end{lemma}
\begin{proof}
    Assume WLOG that $\max_i p_i = p_1$. Then $1/n \leq p_1\leq 1$. Choose any $U_A\in\mathcal{U}(\spa{A})$ such that $\ket{a_1}$ is an eigenvector of $U_A$. Then any term in the double sum of Eqn.~(\ref{eqn:sepstateD}) in which $\ket{a_1}$ appears vanishes. We can hence loosely upper bound the value of Eqn.~(\ref{eqn:sepstateD}) by
    $
        \sqrt{(\sum_{i\neq1,j\neq 1}p_ip_j)}=1-p_1.
    $
    Recalling that $p_1\geq1/n$ yields the desired bound.
\end{proof}

When $n=2$, i.e. when $\rho$ is rank at most two, observe from Lem.~\ref{l:upperbound} that $D(\rho)\leq 1/2$, and this is attainable only when $p_1=p_2=1/2$. We now show that this bound can indeed be saturated, and characterize all states with $n=2$ that do so.

\begin{lemma}
    Let $\rho$ be a separable state as in Eqn.~(\ref{eqn:sepstate}) with $p_1=p_2=1/2$. Then $D(\rho)=1/2$ if and only if $\abs{\braket{a_1}{a_2}}=1/\sqrt{2}$ and $\braket{b_1}{b_2}=0$.
\end{lemma}
\begin{proof}
    Since by Lem.~\ref{l:invariantlocal}, $D(\rho)$ is invariant under local unitaries, we can assume without loss of generality that $\ket{a_1}=\ket{0}$, $\ket{b_1}=\ket{0}$, $\ket{a_2}=\cos\frac{\beta}{2}\ket{0} + \sin\frac{\beta}{2}\ket{1}$ and $\ket{b_2}=\sum_{i=0}^{N-1}\alpha_i\ket{i}$ for $\beta\in[0,\pi]$ and $\alpha_i\in\reals$ with $\sum_i\alpha_i^2=1$, i.e. we can rotate the local states so as to eliminate relative phases. Further, since $\UA\in\rous$ in Eqn.~(\ref{eqn:sepstateD}), we can write $\UA=2\ketbra{u}{u}-I$ for some $\ket{u}=\cos\frac{\theta}{2}\ket{0} + e^{i\phi}\sin\frac{\theta}{2}\ket{1}$, where $\theta,\phi\in[0,2\pi)$. Via the latter, we can rewrite Eqn.~(\ref{eqn:sepstateD}) as:
    \begin{equation}\label{eqn:sepstateD2} \frac{1}{2}\sqrt{\sum_{i,j=1}^2\braket{b_i}{b_j}^2(\braket{a_i}{a_j}^2-\abs{\braket{a_i}{a_j} - 2\braket{a_i}{u}\braket{u}{a_j}}^2)}.
    \end{equation}
    Letting $\Delta$ denote the expression under the square root above, we have by substituting in our expressions for $\ket{a_1}$, $\ket{a_2}$, $\ket{b_1}$, $\ket{b_2}$, and $\ket{u}$ and algebraic manipulation that
    \begin{eqnarray}
        \Delta =\alpha_0^2\left[2\cos\beta
        \sin^2\theta - \sin\beta\sin(2\theta)\cos\phi\right]+\nonumber\\ 1 + \sin^2\theta - (\cos\beta\cos\theta + \sin\beta\sin\theta\cos\phi)^2.\label{eqn:delta}
    \end{eqnarray}
     Our goal is to maximize $\Delta$ with respect to $\alpha_0$ and $\beta$ (which define $\rho$), and then minimize with respect to $\theta$ and $\phi$ (which define $\UA$). Observe now that choosing $\phi=\theta=0$ reduces Eqn.~(\ref{eqn:delta}) to $\Delta = 1-\cos^2\beta$. Hence, unless $\beta=\pi/2$ (i.e. $\abs{\braket{a_1}{a_2}}=1/\sqrt{2}$), we can always achieve $D(\rho)<1/2$. Thus, set $\beta=\pi/2$. Consider next $\phi=0$, and leave $\theta$ unassigned. Then, Eqn.~(\ref{eqn:delta}) reduces to $\Delta = 1 - \alpha_0^2\sin(2\theta)$, from which it is clear that unless $\alpha_0=0$ (i.e. $\braket{b_1}{b_2}=0$), we can always achieve $D(\rho)<1/2$. Plugging these values of $\alpha_0$ and $\beta$ into Eqn.~(\ref{eqn:delta}), we have $\Delta = 1 +\sin^2\theta\sin^2\phi$, from which the claim follows.
\end{proof}
For two-qubit $\rho$, we thus have that with respect to $D(\rho)$ and the geometric discord, the maximally non-classical two qubit states of rank at most two are, up to local unitaries,
\[
    \frac{1}{2}\ketbra{0}{0}\otimes\ketbra{0}{0}+\frac{1}{2}\ketbra{+}{+}\otimes\ketbra{1}{1},
\]
where $\ket{+}=(\ket{0}+\ket{1})/\sqrt{2}$. As mentioned earlier, this matches known results with respect to the relative entropy of quantumness~\cite{GPACH11}. However, the latter analysis is not as general as it begins by with the assumption that $\braket{b_1}{b_2}=0$, whereas we allow arbitrary $\ket{b_1},\ket{b_2}$. It would be interesting to know whether this analysis can be extended to arbitrary rank two-qubit states.

\section{Conclusion} \label{scn:conclusion}We have shown that local unitary operations can indeed form the basis of a faithful non-classicality measure $D(\rho)$ with desirable properties such as: Closed-form expressions for $(2\times N)$-dimensional systems (which coincided with the expression for the geometric discord) and Werner states, a maximum value being attained only for pure maximally entangled states, and faithfulness. We further showed a direct connection between the degeneracy of the spectrum of local unitaries used in our measure and the ability for a state to remain undisturbed under local projective measurements of higher rank. Finally, we gave a characterization of the set of maximally non-classical, yet separable, $(2\times N)$-dimensional $\rho$ of rank at most two (according to $D(\rho)$, and hence also according to the geometric discord).

We leave open the following questions. For what other interesting classes of quantum states can a closed form expression for $D(\rho)$ be found? Can a better intuitive understanding of the interplay between the notions of ``disturbance under local measurements'' and ``disturbance under local unitary operations'' be obtained in higher dimensions? We have given an analytical characterization of all maximally non-classical rank-two $(2\times N)$-dimensional separable states --- we conjecture that higher rank two-qubit states, for example, achieve strictly smaller values of $D(\rho)$. Can this be proven rigorously and analytically? (We remark that a numerical proof for this conjecture was given in~\cite{GA11} for the geometric discord, for example.) What can the study of the generalized notion of quantum discord we defined in Sec.~\ref{scn:discord}, $\delta_{\ve{v}}(\rho)$, tell us about non-classical correlations?

\section*{Acknowledgements}

We thank Gerardo Adesso, Dagmar Bru\ss, Davide Girolami and Marco Piani for helpful discussions. Support from Canada's NSERC, CIFAR and MITACS programs is graciously acknowledged.\\

\noindent \emph{Note:} After completion of this paper, the author learned of independent work in preparation on a similar topic, which has been posted~\cite{SGRBI12} since the first version of the present paper appeared.

\bibliography{Sevmeasure}

\end{document}